\documentclass[doublespacing]{elsart}

\usepackage{epsfig}
\usepackage{amssymb}
\renewcommand{\vec}[1]{\mbox{\boldmath$#1$}}
\usepackage{CJK}

\begin{document}
\begin{frontmatter}
\begin{CJK*}{GB}{gbsn}
\title{The Quantification of Quantum Nonlocality by Characteristic Function}
\author[label1]{Wei Wen (ÎÄΰ)\corauthref{cor1}},
\corauth[cor1]{Corresponding author:chuxiangzi@semi.ac.cn}
\author[label1]{Shu-Shen Li (ÀîÊ÷Éî)}
\address[label1]{State Key Laboratory for Superlattices and Microstructures, Institute of Semiconductors,
Chinese Academy of Sciences, China}


\begin{abstract}
Quantum nonlocal correlation (QNC) is thought to be more general than quantum entanglement correlation, but the strength of it has not been well defined. We propose a way to measure the strength of QNC basing on the characteristic function. The characteristic function of QNC in a composite system is defined as a response function under the local quantum measurement. It is explored that once a characteristic function is given, the state of a composite system, with just a local trace-preserving quantum operation uncertainty, will be determined. We show that the strength of QNC basing on the characteristic function is a half-positive-definite function and does not change under any LU operation. For a two-partite pure state, the strength of QNC is equivalent to the quantum entanglement. Generally, we give a new definition for quantum entanglement using the strength function. Furthermore, we also give a separability-criterion for $2\times m$-dimensional mixed real matrix. This letter proposes an alternate way for QNC further research.
\end{abstract}

\maketitle
\end{CJK*}
\end{frontmatter}











One of the most subtle phenomena in quantum theory is quantum nonlocal correlation (QNC). Although a large amount of research on QNC has been done, it has mainly arisen from the view that nonlocality cannot be described by any local hidden variable theory (LHV). Based on this, some new concepts have been proposed, such as Bell nonlocality\cite{Bell}, quantum entanglement \cite{Bennett,Walgate,Werner}, Schr\"{o}dinger's steerability \cite{Spekkens,Kirkpatrick,steering,entanglement} and so on, which are all defined by different forms of the local joint quantum measurement (LJQM) probability $P(a,b|A,B;W)$ \cite{steering,entanglement}. However, the question about what is the QNC is still far from being solved. Till now, although much attention is paid on the information perspective, the physical aspect of the QNC is also worth studying.

Recently, some other researchers have paid attention to other unorthodox methods and  put forward that nonlocality can be more general \cite{Horodecki,Oxley,Linden,lesspurity,induced}. For example, Bandyopadhyay present that the nonlocality can be redefined by local indistinguishability of a set of orthogonal quantum states, and show that more nonlocality may be with less purity \cite{lesspurity}. Luo and Fu point out that the measurement can induce the nonlocality nonlocality \cite{induced}.
These works all try to find a new way to study quantum nonlocality and inspirit the motivation to reinspect the physical action played in the QNC.

In this Letter, we will further study the quantum nonlocal correlation. We regard the QNC as one sort of elements in the state of a composite system, of which the other element is the set of subsystems. We think the QNC in every special composite system has its own character to be distinguished from others. We express the relation between QNC and the state of a composite system with the mathematical language as follows,
\begin{equation}
\mathrm{\mathbb{B}}:\rho_{ABC\cdots}\rightarrow\{\{\rho_A,\rho_B,\rho_C,\cdots\},\{\mathrm{QNC}s\}\}.
\label{jection}
\end{equation}
The function $\mathbb{B}$ is a bijection and every composite state is mapped into the set of subsystems and the QNC between them. Our task is to find a mathematical expression to describe the QNC
\begin{equation}
\mathrm{\mathbb{B}}':\rho_{ABC\cdots}\rightarrow\{\{\rho_A,\rho_B,\rho_C,\cdots\},\vec{\mathrm{F}}_{QNC}\}.
\label{deft}
\end{equation}
We call the function $\vec{\mathrm{F}}$ that maps the abstract physical quantity $\mathrm{QNC}$ to a  mathematica quantity as a characteristic function of QNC. This is because, just like a special fingerprint corresponds a special man, a special function $\vec{\mathrm{F}}$ corresponds a special composite state once the its sub-states are fixed on. Therefore, we can use it to analysis QNC just like using density matrix to analysis the state of system. Additionally,theoretically speaking, we can redefine the concepts Bell nonlocality, quantum entanglement, Schr\"{o}dinger's steerability and so on basing on the characteristic function. In fact, in the following text, we find a new formulation of quantum entanglement.

Before expatiating on the characteristic function, we will introduce some intuitive-right but unobvious conclusions first. An unlimited quantum measurement is defined as a physical process in which the projection operator can be arbitrarily chosen and the number of copies of the unknown state measured is sufficient. Considering a general situation, an arbitrary projection operator in a finite $n$ Hilbert space is expressed as $\hat{M}(\Theta,\Psi)=|\psi(\Theta,\Psi)\rangle\langle\psi(\Theta,\Psi)|$, where $\Theta$ and $\Phi$ are the sets of variables $\{\psi_k\}$ and $\{\theta_k\}$. $|\psi(\Theta,\Phi)\rangle$ is a pure state in $n$-dimensional space, $|\psi(\Theta,\Psi)\rangle=\sum_{k=1}^{n}a_k|k\rangle$, where $a_k=\prod_{l=1}^{k-1}\sin(\theta_l)\cos(\theta_k)e^{\mathrm{i}\phi_l}$ when $k< n$ and $a_n=\prod_{l=1}^{n-1}\sin(\theta_l)$. Under the unlimited quantum measurement, any unknown state $\rho$ can be distinguished out and this is shown in the following lemma.

\newtheorem{theorem}{Theorem}
\newtheorem{lemma}{lemma}
\begin{lemma}
If $\rho_1$, $\rho_2$ are density matrixes and $\hat{M}(\Theta,\Phi)$ is the projection operator in a finite $n$-dimensional Hilbert space, we say $\rho_1=\rho_2$ when $\mathrm{Tr}(\hat{M}(\Theta,\Phi)\rho_1)=\mathrm{Tr}(\hat{M}(\Theta,\Phi)\rho_2)$ for $\forall \theta_{i}\in (0,\pi)$ and $\forall \phi_{i}\in (0,2\pi)$.
\end{lemma}
\newtheorem{proof}{Proof}
\begin{proof}
It is known any $n\times n$ density matrix $\rho_x$ can be decomposed in an orthonormal basis ${\Gamma_n^{(k)}}$ of traceless generators of group $SU(n)$ \cite{SU1,Book1}
\begin{equation}
  \rho_x=\frac{1}{n}(1+\sum_{k=1}^{n^2-1}r_x^{(k)}\Gamma_n^{(k)}),
  \label{gep}
\end{equation}
where $r_x^{(k)}=\mathrm{Tr}(\rho_x\Gamma_{n}^{(k)})$ and generators $\Gamma_{n}^{(k)}$ have the property of $\mathrm{Tr}(\Gamma_n^{(i)}\Gamma_n^{(j)})=n\delta_{ij}$. The $n^2-1$ real parameters $r_x^{(k)}$ will uniquely determine a density matrix $\rho_{x}$. Therefore, we can use vector $\vec{r}_x=(r_x^{(1)},r_x^{(2)},\ldots,r_x^{(n^2-1)})$ to represent the density matrix. The length of vector $|\vec{r}|=\sqrt{n-1}$.

According to Eq.~(\ref{gep}), when $n=2$, $\rho_1$ and $\rho_2$ can be replaced with vectors $\vec{r}_1$ and $\vec{r}_2$ in the Bloch-Sphere and $\hat{M}$ with $\vec{r}_M(\theta,\phi)$ on the surface of the Bloch-Sphere. Consequently, $\mathrm{Tr}(\hat{M}\rho_1)=\mathrm{Tr}(\hat{M}\rho_2)$ means that $\vec{r}_M(\theta,\phi)\cdot(\vec{r}_1-\vec{r}_2)=0$. The last equation holds iff $\vec{r}_1-\vec{r}_2=0$ , namely $\rho_1=\rho_2$.

For $n> 2$, we can always reduce it to $n(n-1)/2$ $2$-dimensional subsystems and proof any corresponding subsystems equal. To detail this process, we define an operation $V(i,j)_{n,m}=\delta_{n,m}(\delta_{n,i}+\delta_{n,j})$ first. It is known that
\begin{equation}
R_s(i,j)=V(i,j)R V(i,j)^{\dag}
\end{equation}
satisfies $R_{i,j}=(R_s(i,j))_{i,j}$, where $R$ is an $n\times n$ matrix.
Moreover, $\hat{M}_s(i,j)=V(i,j)\hat{M}V(i,j)^{\dag}$ is an equivalent $2$-dimensional projection operator and can be realized by setting some variables $\phi_k$ and $\theta_k$ to zero.  $\mathrm{Tr}(\hat{M}_s(i,j)\rho_1)=\mathrm{Tr}(\hat{M}_s(i,j)\rho_2)$ for every $i$ and $j$, namely  $V(i,j)\rho_1V(i,j)^{\dag}=V(i,j)\rho_2V(i,j)^{\dag}$. Therefore $\rho_1=\rho_2$.
\end{proof}
This theorem is very important in this letter to get the characteristic function. It is shows us that any states, mixed or pure, can be distinguished by unlimited quantum measurements. It is anti-intuition because the general opinion is that $n^2-1$ parameters are needed to fix on an arbitrary $n\times n$ density matrix, but we show here that quantum measurements with $2n-1$ parameters are just enough. Following this, we will give a corollary about the local quantum measurement. It will be shown that unlimit local quantum measurement can also explore the whole information of the composite system.

\newtheorem{corollary}{Corollary}
\begin{corollary}
If $\rho_{AB}$, $\rho_{AB}'$ are density matrixes in $n_A\otimes n_B$-dimensional Hilbert space and $\hat{M}_A(\Theta,\Phi)$ is the projection operator of system $A$, we say $\rho_{AB}=\rho_{AB}'$ when $\mathrm{Tr}_A(\hat{M}_A(\Theta,\Phi)\rho_1)=\mathrm{Tr}_A(\hat{M}(\Theta,\Phi)\rho_2)$ for $\forall \theta_{i}\in (0,\pi)$ and $\forall \phi_{i}\in (0,2\pi)$.
\end{corollary}
\newtheorem{Proof2}{Proof}
\begin{proof}
For a composite system $\mathcal{H}_A\otimes \mathcal {H}_B$, the orthonormal basis $\{\Gamma_{AB}^{k}\}=\{\Gamma_{A}^{i}\otimes \Gamma_{B}^{j}\}$; $n_B\cdot i+j=1,2,\ldots,n_A n_B$. For convenience, $\Gamma^{(0)}=1$ here. Therefore, it is shown that
\begin{eqnarray}
&&\hat{M}_A \otimes I_B=\frac{I_A\otimes I_B}{n_A}+\sum_{i=1}^{n_A}r_M^{(i)}\Gamma_A^{(i)}\otimes I_B;\nonumber \\
&&\rho_{AB}=\frac{I_A\otimes I_B}{n_An_B}+\sum_{j+k=1}^{n_A,n_B}r_{AB}^{(j,k)}\Gamma_A^{(j)}\otimes \Gamma_B^{(k)}.
\end{eqnarray}
It is noted that only these terms $\Gamma_A^{(0)}\Gamma_B^{(i)}$ remain when a partial trace is done under the system $A$, and then we take the form
\begin{eqnarray}
  &&\mathrm{Tr}_A(\hat{M}_A\rho_{AB})-\mathrm{Tr}_A(\hat{M}_A\rho_{AB}')\nonumber\\
  &&=\sum_{k}\vec{r}_M'\cdot \Delta\vec{r}_{AB}^{(k)}\Gamma_B^{(k)},
\end{eqnarray}
where, $\vec{r}_M'=(1,\vec{r}_M)$. Therefore, it is gotten that $\vec{r}_M'\cdot\Delta \vec{r}_{AB}^{(k)}=0$ if $\mathrm{Tr}_A(\hat{M}_A\rho_{AB})-\mathrm{Tr}_A(\hat{M}_A\rho_{AB}')=0$. According to the Lemma 1, $\Delta r_{AB}^{(jk)}=0$, namely $\rho_{AB}=\rho_{AB}'$.
\end{proof}

Let us return to the study of QNC. A local quantum measurement is under a subsystem $\rho_A$ of a composite system $\rho_{AB}$ spanning in the $n_A\times n_B$ Hilbert space. After a local quantum measurement, subsystem $\rho_A$ will collapse into $|\psi(\Theta,\Phi)\rangle_A$, and $\rho_B$ will correspondingly change into
\begin{equation}
\rho_B^M=\frac{\mathrm{Tr}_A(\hat{M}_A\rho_{AB})}{\mathrm{Tr}(\hat{M}_A\rho_{AB})}.
\end{equation}
$\rho_B^M$ is a functional of $\hat{M_A}$. Considering the projection operator has a minimal variety $\hat{M}_A\rightarrow \hat{M}+\delta \hat{M}_A$ (which just like the stimulation input), the sub-state $\rho_B^M$ will correspondingly change into $\rho_B^{M}\rightarrow \rho_B^{M}+\delta\rho_B^M$ (which is the response output). Therefore, we define the following equation about the ratio of change $\delta\rho_B^M/\delta\hat{M}$
\begin{equation}
 \vec{\mathrm{F}}(\Theta,\Psi;\rho_{AB})_{A\rightarrow B}=\sum_{i=1}^{n-1}\frac{\mathrm{Tr}(|\delta^{\theta_i}\rho_B^M|)}{\mathrm{Tr}(|\delta^{\theta_i }\hat{M}_A|)}\vec{e}_{\theta_i}+\frac{\mathrm{Tr}(|\delta^{\phi_i}\rho_B^M|)}{\mathrm{Tr}(|\delta^{\phi_i} \hat{M}_A|)}\vec{e}_{\phi_i}.
  \label{cftion}
\end{equation}
According to Eq.~(3), we get that $\mathrm{Tr}(|\rho_1-\rho_2|)=|\vec{r}_1-\vec{r}_2|$, therefore the equation above can be rewritten as
\begin{equation}
\vec{\mathrm{F}}(\Theta,\Psi;\rho_{AB})_{A\rightarrow B}=\sum_{i=1}^{n-1}\frac{|\delta^{\theta_i}\vec{r}_B^{M}|}{|\delta^{\theta_i}\vec{r}_M|}\vec{e}_{\theta_i}
 +\sum_{i=1}^{n-1}\frac{|\delta^{\phi_i}\vec{r}_B^{M}|}{|\delta^{\phi_i}\vec{r}_M|}\vec{e}_{\theta_i},
\end{equation}
where $\delta^{x}f=(\partial f/\partial x) \delta x$. $\vec{\mathrm{F}}_{A\rightarrow B}$ is a vector in $2n-1$-dimensional space. It forms a surface in this spaces when $\Phi$, and $\Theta$ each change from 0 to $2\phi$. We call $\vec{\mathrm{F}}_{A\rightarrow B}$ the characteristic function of QNC in $\rho_{AB}$ because it is defined totally by the character of QNC. Every special QNC corresponds to a unique characteristic function. This character of $\vec{\mathrm{F}}_{A\rightarrow B}$ can be clearly seen in the following theorem.

\begin{theorem}
Any two density matrixes, $\rho_{AB}^{(i)}$ and $\rho_{AB}^{(j)}$ with the same characteristic function $\vec{\mathrm{F}}_{A\rightarrow B}$ can be transformed into each other by a local unitary transformation under system $B$, namely, $\rho_{AB}^{(i)}=(I_A\otimes U_B^{(ij)})\rho_{AB}^{(j)}(I_A\otimes U_B^{(ij)})^{\dag}$.
\end{theorem}
The proof of the corollary is not complication when one notes that $\mathrm{Tr}(|(\delta\rho_{B}^{M})^{(i)}|)=\mathrm{Tr}(|(\delta\rho_{B}^{M})^{(j)}|)$, which is equivalent to $(\rho_{B}^{M})^{(i)}=U_B^{(ij)}(\rho_{B}^{M})^{(j)}(U_B^{(ij)})^{\dag}$, if $\rho_{AB}^{(i)}$ and $\rho_{AB}^{(j)}$ have a same characteristic function. Basing on the conclusion of corollary 1, we get the theorem above.

The definition of $\vec{\mathrm{F}}_{B\rightarrow A}$ is analogous with the $\vec{\mathrm{F}}_{A\rightarrow B}$ and will not be repeated again. Both $\vec{\mathrm{F}}_{B\rightarrow A}$ and $\vec{\mathrm{F}}_{A\rightarrow B}$  can act as the characteristic function in a composite state. According to this theorem, we can also conclude that $|\vec{\mathrm{F}}(\Theta,\Psi;\rho_{AB}')_{A\rightarrow B}|=|\vec{\mathrm{F}}(\Theta',\Psi';\rho_{AB})_{A\rightarrow B}|$ if $\rho_{AB}'=U_A\otimes U_B \rho_{AB}U_A^{\dag}\otimes U_B^{\dag}$, where $U_A \hat{M}(\Theta,\Psi)U_A^{\dag}=\hat{M}(\Theta',\Psi')$. This is because $\rho_{B}^{M'}=\mathrm{Tr}_A(M_A'\rho_{AB})$, where $M_A'=U_A^{\dag}M_A U_A$. Moreover, $\mathrm{Tr}(|\delta M_A'|)=\mathrm{Tr}(|\delta M_A|)$ because $\delta M_A'=U(\delta Q-\delta S)U^{\dag}$ ( $\delta Q$ and $\delta S$ are infinitesimal positive operators with orthogonal support). Hence $\mathrm{Tr}(|\delta M_A'|)=2\mathrm{Tr}(\delta Q)=\mathrm{Tr}(|\delta M_A|)$.

Let us show examples, for a pure qubit system, $|\psi\rangle_{AB}=\cos\alpha|0,0\rangle+\sin\alpha \exp(\mathrm{i}\gamma)|1,1\rangle$, the characteristic function can be expressed as
\begin{equation}
\vec{\mathrm{F}}(\theta,\phi;|\psi\rangle_{AB})_{A\rightarrow B}=\frac{2|\sin2\alpha|(\vec{e}_\theta+\vec{e}_\phi)}
 {2+\cos2(\theta-\alpha)+\cos2(\theta+\alpha)}.
\end{equation}
To show that local transformation cannot change the shape of $|\vec{\mathrm{F}}|$, we let $|\psi\rangle_{AB}'=U_A|\psi\rangle_{AB}$, where $U_A|0\rangle=\sqrt{3}/2|0\rangle+1/2|1\rangle$, and $U_A|1\rangle=-\sqrt{3}/2|1\rangle+1/2|0\rangle$. We show the pictures of the characteristic function of the state $|\psi\rangle_{AB}$ and $|\psi\rangle_{AB}'$ in the Bloch sphere. As can be seen in Fig.~1, these characteristic functions can be transformed into each other through a rotation.

\begin{figure}
  \includegraphics[width=0.5\textwidth]{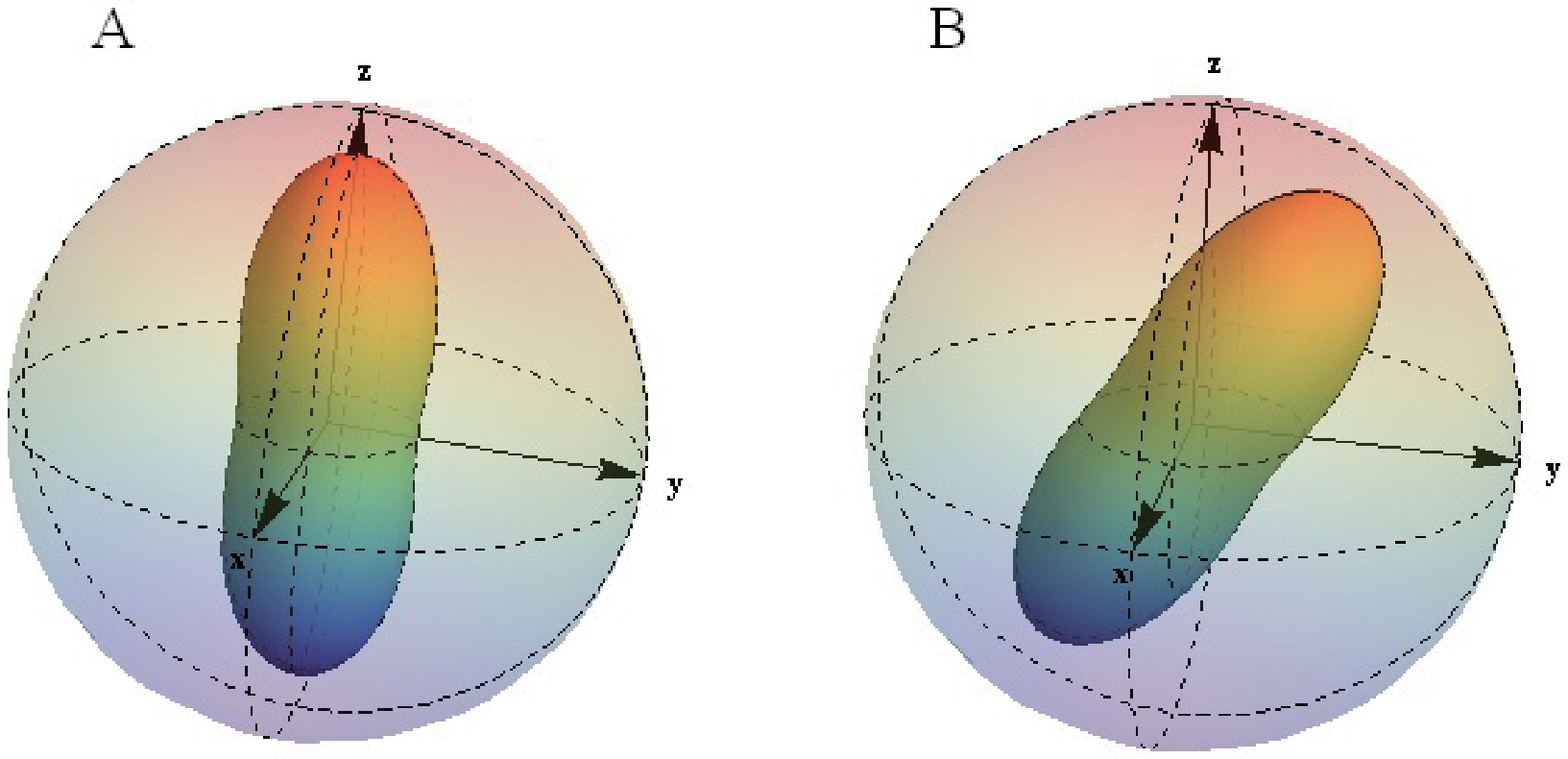}\\
  \caption{The absolute value of the characteristic function in Bloch Sphere. The picture $A$ shows the shape of $|\vec{\mathrm{F}}(\theta,\phi;|\psi\rangle_{AB})_{A\rightarrow B}|$ and the picture $B$ explores the shape of $|\vec{\mathrm{F}}(\theta,\phi;|\psi\rangle_{AB}')_{A\rightarrow B}|$. We choose $\alpha=\pi/3$ here.}\label{Fig1}
\end{figure}

According to this character, we definite a physical quantity that is independent with the form of unlimited quantum measurement as:
\begin{equation}
  G(\rho_{AB})_{A\rightarrow B}=\frac{\sqrt{n_A}}{\Omega}\int_{\Omega} |\vec{\mathrm{F}}_{A\rightarrow B}|\mathrm{Tr}(\rho_A \hat{M}_A)\mathrm{d}_{\mathbb{R}^{2n-2}}\Omega,
  \label{seq}
\end{equation}
where $\mathrm{d}_{\mathbb{R}^{2n-2}}\Omega$ is
\begin{eqnarray}
\mathrm{d}_{\mathbb{R}^{2n-2}}\Omega=\prod_{l=1}^{n-1}\sin(\theta_l)^{n-l-1}\mathrm{d}\theta_l\mathrm{d}\phi_l.
\end{eqnarray}

For a two-particle pure state, it can be proven that $G(\rho_{AB})_{A\rightarrow B}= G(\rho_{AB})_{B\rightarrow A}$, but for an arbitrary state, these two terms are not necessarily equal each other. Hence we define the strength of QNC as the average value of these two terms:
\begin{equation}
G(\rho_{AB})=\frac{1}{2}(G(\rho_{AB})_{A\rightarrow B}+G(\rho_{AB})_{B\rightarrow A}).
\end{equation}

\begin{theorem}
These density matrixes which can be transformed each other by local operation have the same nonlocal strength. Namely, $G(\rho_{AB})=G(U_A\otimes U_B\rho_{AB}U_A^{\dag}\otimes U_B^{\dag})$.
\end{theorem}
Theorem~2 is obvious because according to the analysis above $|\vec{\mathrm{F}}(\Theta,\Psi)'|=|\vec{\mathrm{F}}(\Theta',\Psi')|$ and the integrating range of Eq.~(\ref{seq}) is $SU(2n-2)$ by symmetry.

In terms of the definition of Eq.~(13), some separable states, such as $\rho_{AB}=1/2(|00\rangle\langle00|+|11\rangle\langle11|)$ have a nonlocal correlation although without entanglement (In fact, $G(\rho_{AB})=1/2$ here). It can be also seen that $G(\rho_{AB})$ is not monotonic under LOCC, but it is monotonic under local trace-preserving quantum operation.
\begin{lemma}
Suppose $\varepsilon_p$ is a partial local trace-preserving quantum operation and let $\rho$ be a density operator. Then
\begin{equation}
  G(\rho_{AB})\geq G(\varepsilon_P^x(\rho_{AB})).
\end{equation}
$\varepsilon_p^{A}(\rho_{AB})=\sum_i\lambda_i U_A\otimes I_B \rho_{AB}U_A^{\dag}\otimes I_B$ and  $\varepsilon_p^{B}(\rho_{AB})=\sum_i\lambda_i I_A\otimes U_B \rho_{AB}I_A\otimes U_B^{\dag}$.
\end{lemma}
To prove this lemma, we should use the previous conclusion $D(\rho,\sigma)\geq D(\varepsilon(\rho),\varepsilon(\sigma))$, where $D(\rho,\sigma)$ is the trace distance\cite{Book2}. Basing on this lemma, we will get a more important conclusion as follows.
\begin{corollary}
Suppose $\rho_{AB}=\sum_i{\gamma_i}\rho_{AB}^{(i)}$ is a pure state decomposition of $\rho_{AB}$ and $\rho_{AB}^{\vec{\gamma}\otimes}=\sum_i{\gamma_i}\rho_{A}^{(i)}\otimes\rho_B^{(i)}$, where $\rho_A^{(i)}=\mathrm{Tr}_B(\rho_{AB}^{(i)})$ and $\rho_B^{(i)}=\mathrm{Tr}_A(\rho_{AB}^{(i)})$. Then,
\begin{equation}
G(\rho_{AB})\geq G(\rho_{AB}^{\vec{\gamma}\otimes}).
\label{sgeq}
\end{equation}
\end{corollary}

The mark $\vec{\gamma}\otimes$ here is just an illustration of a direct product decomposition of $\rho_{AB}$ and we call $\rho_{AB}^{\vec{\gamma}\otimes}$ the productization of $\rho_{AB}$ for concision. It is not difficult to understand this corollary because the $\rho_{AB}$ contains more information than $\rho_{AB}^{\vec{\gamma}\otimes}$. We can obtains $\rho_{AB}^{\vec{\gamma}\otimes}$ from $\rho_{AB}$ but not the reverse.
Then, we define a half-positive-definite quantity
\begin{equation}
E(\rho_{AB})=C\inf\{\vec{\gamma}:G(\rho_{AB})-G(\rho_{AB}^{\vec{\gamma}\otimes})\},
\label{eglmt}
\end{equation}
where $\vec{\gamma}=\{\gamma_i,\rho_{AB}^{(i)}\}$ is a symbol of the set of productization. We can determine if $\rho_{AB}$ is separable, $E(\rho_{AB})=0$ and else $E(\rho_{AB})>0$.

Let us look at a special example. The nonlocal correlation strength of a pure qubit state $|\psi\rangle_{AB}=\cos\alpha|0,0\rangle+\sin\alpha \exp(\mathrm{i}\gamma)|1,1\rangle$ is $G(|\psi\rangle)=|\sin 2\alpha|$. Hence, the maximum of strength of QNC appears where the maximum  of entanglement appears. This is not surprising because for an arbitrary two-partite pure state, the form of the productization is determined, equaling to zero. Therefore $E(|\psi_{AB}\rangle)=G(|\psi_{AB}\rangle)$, namely the strength of QNC can be used as the measurement form of quantum entanglement in two-partite pure state.

For a totally mixed state $\rho_{AB}=1/2(|\psi_{+}\rangle\langle\psi_{+}|+|\psi_{-}\rangle\langle\psi_{-}|)$, where $|\psi_{+}\rangle$ and $|\psi_{-}\rangle$ are the Bell states, We get
\begin{equation}
\rho_{AB}^{\vec{\gamma}\otimes}=1/2(|00\rangle\langle00|+|11\rangle\langle11|).
\end{equation}
Therefore, $G(\rho_{AB})=\rho_{AB}^{\vec{\gamma}\otimes}$and the entanglement $E(\rho_{AB})=0$.

We should note that Eq.~(\ref{eglmt}) is usually hard to calculated because we still have not an efficient way to find out the supremum of $G(\rho_{AB}^{\vec{\gamma}\otimes})$ for a general state $\rho_{AB}$. However, it does not mean this definition is useless. Historically, the entanglement of formation had been also hard to calculated initially until the concurrence was proposed. Eq.~(\ref{eglmt}) supports an alternate way to research the quantum entanglement and its values needed further studied. In fact, it is different from previous theories, it is the formulation of the function integral and clearly shows the relationship between QNC and quantum entanglement.

The reason why we take the Eq.~(\ref{deft}) is that we are inspired by the words of Schr\"odinger. In terms of Schr\"odinger's words, in a composite correlated state, a subsystem will be \emph{steered} or \emph{piloted} into one or the other type of state if a local quantum measurement is done on the other subsystem \cite{sdg}. We think this ``steering'' can be seen as the corresponding change of one sub-state when other one be locally measured. It is namely the trace of $\vec{r}_B^{M}$. In fact, the trace of $\vec{r}_{B}^{M}$ will form a surface in a $(n_B^2-1)$-dimensional space when the projection operator $\hat{M}(\Theta,\Phi)$ ranges though the parameter-space $\Theta|=\{[0, \pi]\}$, $\Phi=\{0,2\pi\}$. Studying the surface of $\vec{r}_{B}^{M}$ can result in some new conclusions. For example, consider a $n_A\otimes 2$ composite separate state $\rho_{AB}=\sum_{i}^{m}a_i \rho_A^{(i)}\otimes\rho_B^{(i)}$ and with $\langle\psi_A^{(i)}|\psi_A^{(j)}\rangle=\delta_{ij}$. $\vec{r}_{B}^{M}$ will form a $m$-polyhedron in the Bloch sphere. This is very interesting, because for an inseparable state, $\vec{r}_{B}^{M}$ is usually smooth. In fact, the converse result is also correct under this situation. Namely, if $\vec{r}_{B}^{M}$ forms a $m$-polyhedron, the $\rho_{AB}$ must be expressed by the formula above (To get this conclusion we should use corollary~1). Additionally, a more general conclusion is shown as follows.

\begin{figure}
  \includegraphics[width=0.5\textwidth]{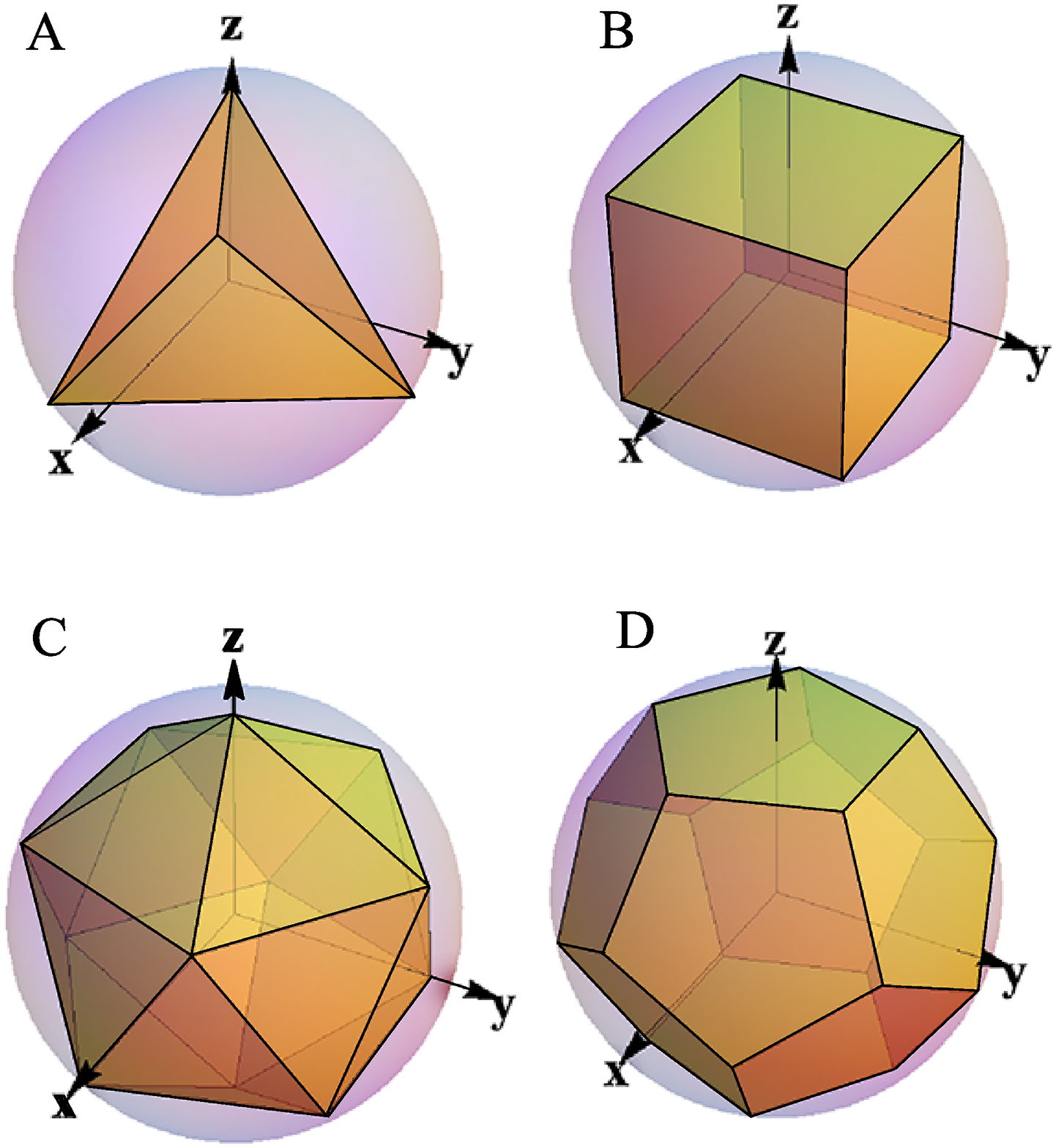}\\
  \caption{The surfaces of $\vec{r}_B^{M}$. We show the $\vec{r}_{B}^{M}$ of $\rho_{AB}=\sum_{i}^{m}a_i \rho_A^{(i)}\otimes\rho_B^{(i)}$ forms the $m$-polyhedron in the Bloch sphere. In this figure, $|\psi_A^{(i)}\rangle=|i\rangle_A$ and $|\psi_B^{(i)}\rangle=\cos(i\pi/m)|0\rangle_B+\sin(i\pi/m)|1\rangle_B$. $m$ are chosen as $4$,$8$,$12$ and $20$ in the picture $A$, $B$, $C$ and $D$ respectively.}\label{Fig1}
\end{figure}

\begin{theorem}
The sufficient and necessary condition for the $2\otimes n_B$ composite state $\rho_{AB}$ to be decomposed into $\rho_{AB}=\sum_{i}a_i\rho_A^{i}\otimes\rho_B^{i}$, where $\rho_A^{i}$ is real density matrix, is that the main normal line of $\mathrm{Tr}(\rho_A\hat{M})\vec{r}_B^{M}$ is constant in $n_B$-dimensional space.
\end{theorem}

\begin{proof}
Let $\vec{r}_b=\mathrm{Tr}(\rho_A\hat{M})\vec{r}_B^{M}=\sum_{k}a_k\lambda_k\vec{r}_B^{(k)}$, where $\lambda_k=\mathrm{Tr}(|\psi_{A}^{k}\rangle\langle\psi_{A}^{k}|\hat{M})$. If $\rho_{AB}=\sum_{i}a_i\rho_A^{i}\otimes\rho_B^{i}$, $\lambda_{k}$ can be expressed as
\begin{equation}
\lambda_k=\cos(\theta-\alpha_A^{k})^2+\sin2\theta\sin2\alpha_A^k\sin(\varphi/2-\varphi_A^k/2)^2.
\end{equation}
Moreover, $\rho_A$ is real, $\varphi_A^{k}=0$, hence
\begin{equation}
\parallel\vec{r}_b^\theta \times \vec{r}_b^\psi\parallel=\parallel\sum_{i,j}\sin2(\alpha_i-\alpha_j)\vec{r}_B^{(i)}\times\vec{r}_B^{(j)}\parallel,
\end{equation}
where $\parallel\vec{r}\parallel$ means the normalization of $\vec{r}$. $\vec{r}_B^{(k)}$ is independent of $\theta$ and $\phi$, therefore $\parallel\vec{r}_b^\theta \times \vec{r}_b^\psi\parallel=cons.$  Namely, the direction of the main normal line does not change with the variables $\phi$ and $\theta$. Conversely, if the main normal line $\vec{r}_{\mathbf{n}}$ is constant, it can be rewritten as $\vec{r}_{\mathbf{n}}=\parallel\sum_{i,j}\sin2(\alpha_i-\alpha_j)\vec{r}_B^{(i)}\times\vec{r}_B^{(j)}\parallel $ when $\alpha_i$ and $\alpha_j$ are appropriately chosen. Therefore, $\vec{r}_b$ can be determined. consequently, according to corollary~1,  $\rho_{AB}$ is a separable state.
\end{proof}

\section{Discussion and conclusion} In this above work, we calculate the characteristic function and the strength of QNC. We regard the characteristic function as a corresponding function and through the trace distance, differential theory is brought into QNC research. It supports an alternate way to analysis and categorize the QNC. In fact, as a special sort of QNC, quantum entanglement is defined and shown in Eq.~(16). This definition is shown a more definite and clear relationship between the quantum entanglement and QNC. We also show a way to distinguish whether a real mixed composite density matrix is separable or not.

However, it should be reminded that the monotonic character of our new definition of quantum entanglement under LOCC has not rigorously proven yet although we have done some computation to show that it is correct. The productization density matrix $\rho_{AB}^{\vec{\gamma}\otimes}$ is brought into this Letter. It seems that the supremum of it only possesses the ``local'' correlation of $\rho_{AB}$. It is can be seen that the definition $E_s(\rho_{AB})=\mathrm{inf}\{\vec{\gamma}:S(\rho_{AB}^{\vec{\gamma}\otimes})-S(\rho_{AB})\}$,
is also a possible entanglement measurement definition, where $S(\rho_{AB})$ is the Shannon entropy of $\rho_{AB}$ here. This work was supported by the National Basic Research Program of China (973 Program) grant No. G2009CB929300.

\end{document}